\documentclass[journal]{IEEEtran}

\usepackage[utf8]{inputenc}
\usepackage[cmex10]{amsmath}
\usepackage{amssymb}
\usepackage{amsthm}
\usepackage{graphicx}
\usepackage{epstopdf}

\DeclareMathSymbol{\ast}{\mathbin}{symbols}{"03}

\newtheorem{theorem}{Theorem}
\newtheorem{proposition}{Proposition}

\newtheorem{lemma}{Lemma}
\newtheorem{definition}{Definition}
\newtheorem{example}{Example}
\newtheorem{remark}{Remark}
\newenvironment{theorembis}[1]
  {\renewcommand{\thetheorem}{\ref{#1}$'$}%
   \addtocounter{theorem}{-1}%
   \begin{theorem}}
  {\end{theorem}}

\newcommand{\defn}{\mathrel{\stackrel{\Delta}{=}}}

\def\bU{\mathbf{U}}
\def\bx{\mathbf{x}}
\def\bX{\mathbf{X}}
\def\by{\mathbf{y}}
\def\bY{\mathbf{Y}}

\def\bS{\mathbf{S}}
\def\bT{\mathbf{T}}
\def\bA{\mathbf{A}}
\def\bB{\mathbf{B}}
\def\bG{\mathbf{G}}
\def\bF{\mathbf{F}}

\def\bFo{\mathbf{F}_{\text{out}}}
\def\bFi{\mathbf{F}_{\text{in}}}

\def\cA{\mathcal{A}}
\def\cB{\mathcal{B}}
\def\cX{\mathcal{X}}
\def\cY{\mathcal{Y}}

\def\cM{\mathcal{M}}
\def\cH{\mathcal{H}}

\DeclareMathOperator{\var}{Var}
\DeclareMathOperator{\varent}{V}
\DeclareMathOperator{\cov}{Cov}

\def\Hm{g_1}
\def\Hp{g_2}

\def\hip{h_{\text{in},2}}
\def\him{h_{\text{in},1}}
\def\hp{h_{\text{out},2}}
\def\hm{h_{\text{out},1}}
\def\hoi{h_{\text{out},i}}
\def\Aim{A_{\text{in},1}}
\def\Aip{A_{\text{in},2}}
\def\cAm{\overline{A}_{\text{out},1}}

\def\cAim{\overline{A}_{\text{in},1}}
\def\cAip{\overline{A}_{\text{in},2}}
\def\Am{A_{\text{out},1}}
\def\Ap{A_{\text{out},2}}
\def\Bim{B_{\text{in},1}}
\def\Bip{B_{\text{in},2}}
\def\Bm{B_{\text{out},1}}
\def\Bp{B_{\text{out},2}}
\def\Aii{A_{\text{in},i}}
\def\Bii{B_{\text{in},i}}
\def\Aoi{A_{\text{out},i}}
\def\Boi{B_{\text{out},i}}
\def\cBim{\overline{B}_{\text{in},1}}

\def\aim{a_{\text{in},1}}
\def\aip{a_{\text{in},2}}

\def\am{a_{\text{out},1}}

\def\aii{a_{\text{in},i}}
\def\bii{b_{\text{in},i}}

\def\Fim{F_{\text{in},1}}
\def\Fom{F_{\text{out},1}}
\def\Fip{F_{\text{in},2}}
\def\Fop{F_{\text{out},2}}

\def\R{\mathbb{R}}
\def\E{\mathbb{E}}

\newcommand*\diff{\mathop{}\!\mathrm{d}}

\begin{document}

\title{Varentropy Decreases Under the Polar Transform}
\author{Erdal Ar{\i}kan
\thanks{The author is with the Department of Electrical-Electronics Engineering,  
Bilkent University, 06800, Ankara, Turkey. Email: arikan@ee.bilkent.edu.tr}}

\maketitle

\begin{abstract}
We consider the evolution of variance of entropy (varentropy) in the course of a polar transform operation on binary data elements (BDEs).
A BDE is a pair $(X,Y)$ consisting of a binary random variable $X$ and an arbitrary side information random variable $Y$. 
The varentropy of $(X,Y)$ is defined as the variance of the random variable $-\log p_{X|Y}(X|Y)$.
A polar transform of order two is a certain mapping that takes two independent BDEs and produces two new BDEs that are correlated with each other.
It is shown that the sum of the varentropies at the output of the polar transform is less than or equal to the sum of the varentropies at the 
input, with equality if and only if at least one of the inputs has zero varentropy.
This result is extended to polar transforms of higher orders and it is shown that the varentropy decreases to zero asymptotically when the BDEs at the input are independent
and identically distributed.
\end{abstract}


\section{Introduction}\label{Sect:Introduction}

We use the term ``varentropy'' as an abbreviation for ``variance of the conditional entropy random variable'' following the usage in \cite{KV}.
In his pioneering work, Strassen \cite{Strassen} showed that the varentropy is a key parameter for estimating the performance of optimal block-coding schemes 
at finite (non-asymptotic) block-lengths.
More recently, the comprehensive work by Polyanskiy, Poor and Verd\' u \cite{PPV} further elucidated the significance of varentropy (under the name ``dispersion'') and rekindled interest in the subject.
In this paper, we study varentropy in the context of polar coding.
Specifically, we track the evolution of average varentropy in the course of polar transformation of independent identically distributed (i.i.d.) BDEs 
and show that it decreases to zero asymptotically as the transform size increases.
As a side result, we obtain an alternative derivation of the polarization results of \cite{ArikanIT2009}, \cite{ArikanSource2010}.

\subsection{Notation and basic definitions}
Our setting will be that of binary-input memoryless channels and binary memoryless sources.
We treat source and channel coding problems in a common framework by using the neutral term ``binary data element'' (BDE)
to cover both.
Formally, a BDE is any pair of random variables $(X,Y)$ where $X$ takes values over $\cX\defn \{0,1\}$ (not necessarily from the uniform distribution) 
and $Y$ takes values over some alphabet $\cY$ which may be discrete or continuous.
A BDE $(X,Y)$ may represent, in a source-coding setting, a binary data source $X$ that we wish to compress in the presence of some side information $Y$; 
or, it may represent, in a channel-coding setting, a channel with input $X$ and output $Y$.

Given a BDE $(X,Y)$, the information measures of interest in the sequel will be the {\sl conditional entropy random variable} 
$$
h(X|Y) \defn -\log p_{X|Y}(X|Y),
$$
the {\sl conditional entropy}
$$
H(X|Y) \defn \E \,h(X|Y),
$$
and, the {\sl varentropy}
$$
\varent(X|Y) \defn \var (h(X|Y)).
$$
Throughout the paper, we use base-two logarithms.

The term {\sl polar transform} is used in this paper to to refer to an operation that takes two {\sl independent} BDEs $(X_1,Y_1)$ and $(X_2,Y_2)$ as input, and produces
two new BDEs $(U_1,\bY)$ and $(U_2;U_1,\bY)$ as output, where $U_1 \defn X_1 \oplus X_2$, $U_2 \defn X_2$, and $\bY \defn (Y_1,Y_2)$.
The notation ``$\oplus$'' denotes modulo-2 addition.

\subsection{Polar transform and varentropy}\label{sec:PolarTransform}

The main result of the paper is the following. 
\begin{theorem}\label{theoremvarentropy}
The varentropy is nonincreasing under the polar transform in the sense that,
if $(X_1,Y_2)$, $(X_2,Y_2)$ are any two independent BDEs at the input of the transform and $(U_1,\bY)$, $(U_2;U_1,\bY)$ are the BDEs at its output, then
\begin{align}\label{varentropycontraction0}
\varent(U_1|\bY) + \varent(U_2|U_1,\bY) \le \varent(X_1|Y_1) + \varent(X_2|Y_2),
\end{align}
with equality if and only if (iff) either $V(X_1|Y_1)=0$ or $V(X_2|Y_2)=0$.
\end{theorem}

For an alternative formulation of the main result,
let us introduce the following notation:
\begin{align}\label{inputentropies}
\him \defn h(X_1|Y_1), \qquad \hip \defn h(X_2|Y_2),
\end{align}
\begin{align}\label{outputentropies}
\hm \defn h(U_1|\bY), \qquad \hp \defn h(U_2|U_1,\bY).
\end{align}

Theorem~\ref{theoremvarentropy} can be reformulated as follows.
\begin{theorembis}{theoremvarentropy}\label{theoremvarentropy2}
The polar transform of conditional entropy random variables, $(\him,\hip)\to (\hm,\hp)$, produces positively correlated output entropy terms
in the sense that
\begin{align}\label{covarianceinequality}
\cov(\hm,\hp)\ge 0,
\end{align}
with equality iff either $\var(\him)=0$ or $\var(\hip)=0$.
\end{theorembis}
This second form makes it clear that any reduction in varentropy can be attributed entirely to the creation of a positive correlation between the entropy random variables $\hm$ and $\hp$ at the output of the polar transform.

Showing the equivalence of the two claims \eqref{varentropycontraction0} and \eqref{covarianceinequality} is a simple exercise.
We have, by the chain rule of entropy, 
\begin{equation}\label{conservation1}
\hm + \hp = \him + \hip;
\end{equation}
hence,
$\var(\hm + \hp) = \var(\him + \hip)$.
Since $\him$ and $\hip$ are independent,
$\var(\him + \hip) = \var(\him) + \var(\hip)$;
while
$\var(\hm + \hp)  =\var(\hm) + \var(\hp)+ 2\cov(\hm,\hp)$.
Thus, the claim \eqref{varentropycontraction0}, which can be written in the equivalent form
\begin{align*}
\var(\hm) + \var(\hp) \le \var(\him) + \var(\hip),
\end{align*}
is true iff \eqref{covarianceinequality} holds.

A technical question that arises in the sequel is whether the varentropy is uniformly bounded across the class of all BDEs.
This is indeed the case.
\begin{lemma}\label{lemma:bound}
For any BDE $(X,Y)$, $V(X|Y)\le 2.2434$.
\end{lemma}
\begin{proof}
It suffices to show that the second moment of $h(X|Y)$ satisfies the given bound.
\begin{align*}
E[h(X|&Y)^2]  \le \max_{0\le x\le 1} [x\log^2(x)+(1-x)\log^2(1-x)]\\
& \le 2\max_{0\le x\le 1} [x\log^2(x)] = 8 e^{-2}\log^2(e) \approx 2.2434.
\end{align*}
(A numerical study shows that a more accurate bound on $V(X|Y)$ is $1.1716$, but the present bound will be sufficient for our purposes.) 
\end{proof}
This bound guarantees that all varentropy terms in this paper exist and are bounded;
it also guarantees the existence of the covariance terms since by the Cauchy-Schwarz inequality we have $|\cov(\hm,\hp)|\le \sqrt{\var(\hm)\var(\hp)}$.

We will end this part by giving two examples in order to illustrate the behavior of varentropy under the polar transform.
The terminology in both examples reflects a channel coding viewpoint; although, each model may also arise in a source coding context.
\begin{example}\label{examplebsc}
In this example, $(X,Y)$ models a binary symmetric channel (BSC) with equiprobable inputs and a crossover probability $0\le \epsilon\le 1/2$;
in other words, $X$ and $Y$ take values in the set $\{0,1\}$ with 
$$
p_{X,Y}(x,y) = \begin{cases} \frac12 (1-\epsilon), & \text{if $x=y$};\\
\frac12 \epsilon, & \text{if $x\neq y$}.
\end{cases}
$$
Fig.~\ref{figurebscvarentropy} gives a sketch of the varentropy and covariance terms defined above,
with $\var(h_{\text{in}})$ denoting the common value of $\var(\him)$ and $\var(\hip)$). 
(Formulas for computing the varentropy terms will be given later in the paper.)
The non-negativity of the covariance is an indication that the varentropy is reduced by the polar transform.
\begin{figure}[ht]
\begin{center}
\resizebox{!}{6cm}{%
\includegraphics{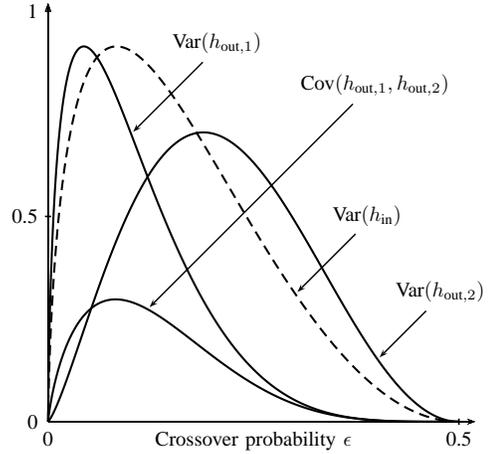}
}
\caption{Variance and covariance of entropy for BSC under polar transform.}
\label{figurebscvarentropy}
\end{center}
\end{figure}
\end{example}

\begin{example}\label{examplebec}
Here, $(X,Y)$ represents a binary erasure channel (BEC) with equiprobable inputs and an erasure probability $\epsilon$.
In other words, $X$ takes values in $\{0,1\}$, $Y$ takes values in $\{0,1,2\}$, and
$$
p_{X,Y}(x,y) = \begin{cases} \frac12 (1-\epsilon), & \text{if $x=y$};\\
\frac12 \epsilon, & \text{if $y=2$}.
\end{cases}
$$
In this case, there exist simple formulas for the varentropies.
$
\var(\him) = \var(\hip) = \var(h_{\text{in}}) = \epsilon(1-\epsilon),
$
$
\var(\hp) = (2\epsilon- \epsilon^2)(1-\epsilon)^2,
$
$
\var(\hm) = \epsilon^2(1-\epsilon^2).
$
The covariance is given by $\cov(\hm,\hp) = \epsilon^2(1-\epsilon)^2$.
The corresponding curves are plotted in Fig.~\ref{figurebecvarentropy}.

\begin{figure}[th]
\begin{center}
\resizebox{!}{6cm}{%
\includegraphics{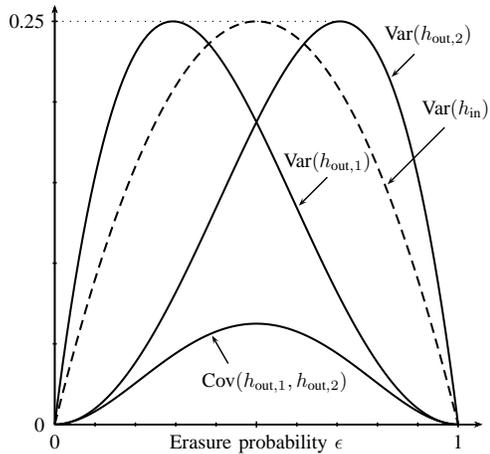}
}
\caption{Variance and covariance of entropy for BEC under polar transform.}
\label{figurebecvarentropy}
\end{center}
\end{figure}

\end{example}

\subsection{Organization}
The rest of the paper is organized as follows.
In Section~\ref{representations}, we define two canonical representations for a BDE $(X,Y)$
that eliminate irrelevant details from problem description and simplify the analysis.
In Section~\ref{Sect:correlation}, we review some basic facts about the covariance function that are needed in the remainder of the paper.
Section~\ref{sectionproof} contains the proof of Theorem~\ref{theoremvarentropy2}.
Section~\ref{higherorder} considers the behavior of varentropy under higher-order polar transforms and contains a self-contained proof of the main polarization result of \cite{ArikanIT2009}.

Throughout, we will often write $\overline{p}$ to denote $1-p$ for a real number $0\le p\le 1$. For $0\le p,q\le 1$, we will write 
$p \ast q$ to denote the convolution $pq + \overline{p}\,\overline{q}$.

\section{Canonical representations}\label{representations}
The information measures of interest relating to a given BDE $(X,Y)$ are determined solely by the joint probability distribution of $(X,Y)$; the specific forms of the alphabets $\cX$ and $\cY$ play no role. We have already fixed $\cX$ as $\{0,1\}$ so as to have a standard representation for $X$.
It is possible and desirable to re-parametrize the problem, if necessary, so that $\cY$ also has a canonical form.
Such canonical representations have been given for Binary Memoryless Symmetric (BMS) channels in \cite{MCT}. 
The class of BDEs $(X,Y)$ under consideration here is more general than the class of BMS channels, but similar ideas apply.
We will give two canonical representations for BDEs, which we will call the $\alpha$-representation and the $\beta$-representation.
The $\alpha$-representation replaces $\cY$ with a canonical alphabet $\cA\subset [0,1]$, and has the property of being ``lossless''.
The $\beta$-representation replaces $\cY$ with $\cB \subset [0,1/2]$; it is ``lossy'', but happens to be more convenient than the $\alpha$-representation
for purposes of proving Theorem~\ref{theoremvarentropy2}. 

\subsection{The $\alpha$-representation}\label{Subsect:alpha}

Given a BDE $(X,Y)$, we associate to each $y\in \cY$ the parameter
\begin{equation*}
\alpha(y) =\alpha_{X|Y}(y) \defn p_{X|Y}(0|y)
\end{equation*}
and define $A\defn \alpha(Y)$. The random variable $A$ takes values in the set $\cA \defn \{\alpha(y): y\in \cY\}$, which is always a subset of $[0,1]$. 
We refer to $A$ as the $\alpha$-representation of $(X,Y)$. 
The $\alpha$-representation provides economy by using a canonical alphabet $\cA$ in which any two symbols $y,y'\in \cY$ are merged into a common symbol $a$ whenever $\alpha(y)=\alpha(y')=a$.

We give some examples to illustrate the $\alpha$-representation.
For the BSC of Example~\ref{examplebsc}, we have $\alpha(0)=1-\epsilon$, $\alpha(1)=\epsilon$, $\cA = \{\epsilon,1-\epsilon\}$.
In the case of the BEC of Example~\ref{examplebec}, we have $\alpha(0)=1$, $\alpha(1)=0$, $\alpha(2)=1/2$, $\cA = \{0,1/2,1\}$.
As a third example, consider the channel $y= (-1)^x c + z$ where $c>0$ is a constant and $z\sim N(0,1)$ is a zero-mean unit-variance additive Gaussian noise,
independent of $x$.
In this case, we have
$$
\alpha(y) = \frac{e^{-(y-c)^2/2}}{e^{-(y-c)^2/2}+e^{-(y+c)^2/2}} = \frac{1}{1+e^{-2cy}},
$$
giving $\cA=(0,1)$.

The $\alpha$-representation provides ``sufficient statistics'' for computing the information measures of interest to us.
To illustrate this, let $(X,Y)$ be an arbitrary BDE and let $A= \alpha(Y)$ be its $\alpha$-representation.
Let $F_A$ denote the cumulative distribution function (CDF) of $A$.

The conditional entropy random variable is given by 
\begin{equation}\label{eq:entropyA}
h(X|Y) = h(X|A) = \begin{cases} - \log A, & X=0;\\
-\log \overline{A}, & X=1.
\end{cases}
\end{equation}
Hence, the conditional entropy can be calculated as 
\begin{align}
H(X|Y) & = \E\, h(X|Y) = \E\, h(X|A) =\E_A \E_{X|A}\, h(X|A)\nonumber\\
& = \E_A \cH(A) =\E\, \cH(A) = \int_0^1 \cH(a) \diff F_A(a), \label{eq:H(X|A)}
\end{align}
where $\cH(a) \defn -a \log a -\overline{a}\log \overline{a}$, $a\in [0,1]$, is the binary entropy function.
Likewise, the varentropy is given by 
\begin{equation}
V(X|Y) = V(X|A) = \E\, \cH_2(A) - \bigl[\E\, \cH(A)\bigr]^2 \label{eq:varentropyA},
\end{equation}
where $\cH_2(a) \defn -a \log^2 a -\overline{a} \log^2 \overline{a}$ and
\begin{align*}
\E\,\cH_2(A) = \int_0^1 \cH_2(a) \diff F_A(a).
\end{align*}

Finally, we note that $H(X) = \cH(p_X(0)) = \cH(\E\,A)$.
Thus, all information measures of interest in this paper can be computed given knowledge of the distribution of $A$.
\subsection{The $\beta$-representation}\label{Subsect:beta}
Although the $\alpha$-representation eliminates much of the irrelevant detail from $(X,Y)$, there is need for an even more compact representation for the type of problems considered in the sequel.
This more compact representation is obtained by associating to each $y\in \cY$ the parameter
\begin{align*}
\beta(y) & =\beta_{X|Y}(y) \defn \min\{p_{X|Y}(0|y),p_{X|Y}(1|y)\}.
\end{align*}
We define the $\beta$-representation of $(X,Y)$ as the random variable $B\defn \beta(Y)$. We denote the range of $B$ by $\cB\defn \{\beta(y): y\in \cY\}$ and note that $\cB \subset [0,1/2]$.

The $\beta$-representation can be obtained from the $\alpha$-representation by
$$
\beta(y) = \min\{\alpha(y),1-\alpha(y)\}, \quad B = \min\{A,\overline{A}\,\};
$$
but, in general, the $\alpha$-representation cannot be recovered from the $\beta$-representation.

For the BSC of Example~\ref{examplebsc}, we have $\beta(0)=\beta(1) =\epsilon$, giving $\cB = \{\epsilon\}$.
For the BEC of Example~\ref{examplebec}, we have $\beta(0)=\beta(1)=0$, $\beta(2)=1/2$, and $\cB = \{0,1/2\}$.
For the binary-input additive Gaussian noise channel, we have 
$$
\beta(y) = \frac{1}{1+e^{2c|y|}},
$$
with $\cB=(0,1/2]$.

As it is evident from \eqref{eq:entropyA}, the conditional entropy random variable $h(X|Y)$ cannot be expressed as a function of $(X,B)$.
However, if the CDF $F_B$ of $B$ is known, we can compute $H(X|Y)$ and $\varent(X|Y)$ by the following formulas that are analogous to \eqref{eq:H(X|A)} and \eqref{eq:varentropyA}:
\begin{equation*}
H(X|Y) = \E\, \cH(B), \quad V(X|Y) = \E\, \cH_2(B) - \bigl[\E\,\cH(B)\bigr]^2.
\end{equation*}

To see that $B$ is less than a ``sufficient statistic'' for information measures, one may note that $H(X)$ is not determined by knowledge of $F_B$ alone.
For example, for a BDE $(X,Y)$ with $\Pr(Y=X)=1$, we have $\Pr(B=0)=1$, independently of $p_X(0)$. 

Despite its shortcomings, the $\beta$-representation will be useful for our purposes due to the fact that the binary entropy function $\cH(p)$ is monotone over $p\in [0,1/2]$   
but not over $p\in [0,1]$. Thus, the random variable $\cH(B)$ is a monotone function of $B$ over the range of $B$, but $\cH(A)$ is not necessary so over the range of $A$.
This monotonicity will be important in proving certain correlation inequalities later in the paper.

\subsection{Classification of binary data elements}
Table~\ref{table:classification} gives a classification of a BDE $(X,Y)$ in terms of the properties of $B=\beta(Y)$. 
The classification allows an erasing BDE to be extreme as a special case.

\begin{table}[!ht]
\renewcommand{\arraystretch}{1.2}
\caption{Classification of BDEs}
\centering
\begin{tabular}{l|l}
\hline
Type & Property\\\hline\hline
pure &    $P(B = b)=1$ for some $b\in [0,1/2]$\\
extreme &  $P(B= 0) =1$ or $P(B=1/2) = 1$\\
perfect &   $P(B= 0) =1$\\
purely random (p.r.) &  $P(B=1/2) =1$\\
erasing & $P(B=0) + P(B=1/2) = 1$\\\hline
\end{tabular}
\renewcommand{\arraystretch}{1}
\label{table:classification}
\end{table}
\vspace*{3mm}

For a pure $(X,Y)$, we obtain from \eqref{eq:H(X|A)} and \eqref{eq:varentropyA} that
$$
H(X|Y) = \cH(b), \quad
V(X|Y) =  b(1-b) \log^2\bigg(\frac{b}{1-b}\bigg),
$$
where $b$ is the value that $B=\beta(Y)$ takes with probability 1.
A simple corollary to this is the following characterization of an extreme BDE. 
\begin{proposition}\label{prop:extreme}
Let $(X,Y)$ be a BDE and $B=\beta(Y)$.
The following three statements are equivalent: (i) $(X,Y)$ is extreme, (ii) $H(X|Y)=0$ or $H(X|Y)=1$, (iii) $\varent(X|Y)=0$. 
\end{proposition}
We omit the proof since it is immediate from the above formulas for $H(X|Y)$ and $V(X|Y)$ for a pure BDE.

For an erasing $(X,Y)$, it is easily seen that 
$$
H(X|Y) = p, \quad V(X|Y) = p(1-p)
$$
where $p=P[\beta(Y)=1/2]$ is the {\sl erasure} probability.

Parenthetically, we note that while the entropy function satisfies $H(X|Y)\le H(X)$, there is no such general relationship between $V(X|Y)$ and $V(X)$.
For an erasing $(X,Y)$ with $p_X(1)=1-p_X(0)=q$ and erasure probability $p$, we have $V(X) = q(1-q) \log^2[q/(1-q)]$ while $V(X|Y)=p(1-p)$.
Either $V(X) < V(X|Y)$ or $V(X) > V(X|Y)$ is possible depending on $q$ and $p$.

\subsection{Canonical representations under polar transform}
In this part, we explore how the $\alpha$- and $\beta$-representations evolve as they undergo a polar transform.
Let us return to the setting of Sect.~\ref{sec:PolarTransform}. Let $(U_1,\bY)$ and $(U_2; U_1,\bY)$ denote the two BDEs obtained from a pair of independent BDEs $(X_1,Y_1)$ and $(X_2,Y_2)$ by the polar transform.
Let $\him$, $\hip$, $\hm$, and $\hp$ denote the entropy random variables at the input and output of the polar transform.
For $i=1,2$, let $\Aii$ and $\Bii$ be the $\alpha$- and $\beta$-representations for the $i$th BDE at the input side;
and let $\Aoi$ and $\Boi$ be those for the $i$th BDE at the output side.
Let the sample values of these variables be denoted by small-case letters, such as $\aii$ for $\Aii$, $\bii$ for $\Bii$, etc.

\begin{proposition}\label{prop:alpha2alpha}
The $\alpha$-parameters at the input and output of a polar transform are related by 
\begin{align}\label{eq:rel1}
\Am = \Aim\ast \Aip,
\end{align}
\begin{align}\label{eq:rel2}
\Ap & = \begin{cases} \Aim \Aip/(\Aim\ast \Aip), & U_1=0;\\
\cAim \Aip/(\cAim\ast\Aip), & U_1 = 1.
\end{cases}
\end{align}
\end{proposition}

\begin{remark}
In \eqref{eq:rel2}, the event $\{\Aim\ast \Aip = 0\}$ leads to an indeterminate form $\Ap = 0/0$, but the conditional probability of $\{\Aim\ast \Aip = 0\}$ given $\{U_1=0\}$ is zero:
$\Aim\ast \Aip = 0$ implies $(\Aim,\Aip)\in \{(0,1),(1,0)\}$, which in turn implies $(X_1,X_2)\in \{(1,0),(0,1)\}$, giving $U_1 = 1$. Similarly, the event $\{\cAim\ast \Aip = 0\}$ is incompatible with $\{U_1=1\}$.  
\end{remark}

\begin{proof}
For a fixed $\bY=(y_1,y_2)$, the sample values of $\Am$ are given by
\begin{align*}
\am(y_1,y_2) & \defn p_{U_1|Y_1,Y_2}(0|y_1,y_2)\\
 & =\sum_{u_2} p_{U_1,U_2|Y_1,Y_2}(0,u_2|y_1,y_2)\\
& = \sum_{u_2} p_{X_1|Y_1}(u_2|y_1)p_{X_2|Y_2}(u_2|y_2)\\
& = \aim(y_1)\ast \aip(y_2).
\end{align*}
From this, the first statement \eqref{eq:rel1} follows.
The second statement \eqref{eq:rel2} can be obtained by similar reasoning.
\end{proof}

The above result leads to the following ``density evolution'' formula.
Let $\Fim$, $\Fip$, $\Fom$, and $\Fop$ be the CDFs of $\Aim$, $\Aip$, $\Am$, and $\Ap$, respectively.
\begin{proposition}\label{prop:densityevolution}
The CDFs of the $\alpha$-parameters at the output of a polar transform are related to the CDFs of the $\alpha$-parameters at the input by 
\begin{align*}
\Fom(a) & = \iint\limits_{a_1\ast a_2 \le a} \diff \Fim(a_1) \diff \Fip(a_2)
\end{align*}
\begin{align*}
\Fop(a) & = \iint\limits_{(a_1a_2/a_1\ast a_2) \le a} (a_1\ast a_2) \diff \Fim(a_1) \diff \Fip(a_2)\\
& + \iint\limits_{(\overline{a}_1a_2/\overline{a}_1\ast a_2) \le a} (\overline{a}_1\ast a_2) \diff \Fim(a_1) \diff \Fip(a_2)
\end{align*}
\end{proposition}
These density evolution equations follow from \eqref{eq:rel1} and \eqref{eq:rel2}. In the expression for $\Fop(a)$, the integrands $(a_1\ast a_2)$ and $(\overline{a}_1\ast a_2)$ correspond to the conditional probability of $U_1$ being 0 and 1, respectively, given that $\Aim=a_1$ and $\Aip=a_2$. We omit the proof for brevity.

For the $\beta$-parameters, the analogous result to Proposition~\ref{prop:alpha2alpha} is as follows.
\begin{gather*}
\Bm = \gamma(\Bim\ast \Bip),\\
\Bp  = \begin{cases} \gamma(\Bim\Bip/(\Bim\ast \Bip)), & \Gamma >0;\\
 \gamma(\cBim \Bip/(\cBim\ast \Bip)), & \Gamma \le 0,
\end{cases}
\end{gather*}
where $\gamma(x) \defn \min\{x,1-x\}$ for any $x\in [0,1]$ and $\Gamma \defn (1/2 - U_1)(1/2 -\Aim)(1/2 -\Aip)$.
We omit the derivation of these evolution formulas for the $\beta$-parameters since they will not be used in the sequel.
The main point to note here is that the knowledge of $(\Bim,\Bip,U_1)$ is not sufficient to determine $\Gamma$,
hence not sufficient to determine $\Bp$. 
So, there is no counterpart of Proposition~\ref{prop:densityevolution} for the $\beta$-parameters.

Although there is no general formula for tracking the evolution of the $\beta$-parameters through the polar transform,
there is an important exceptional case in which we can track that evolution, namely, the case where at least one of the BDEs at the transform input is extreme. 
This special case will be important in the sequel, hence we consider it in some detail.

Table~\ref{table:extremepolar} summarizes the evolution of the $\beta$-parameters for all possible situations in which at least one of the input BDEs is extreme.
(In the table ``p.r.'' stands for ``purely random''.)

\begin{table}[!ht]
\renewcommand{\arraystretch}{1.2}
\caption{Polar transform of extreme BDEs}
\centering
\begin{tabular}{c|c|c|c}
\hline
$\Bim$ & $\Bip$ & $\Bm$  & $\Bp$ \\\hline\hline
perfect &  any  & $\Bip$ & perfect\\
p.r. &  any  & p.r. & $\Bip$\\
any &  perfect & $\Bim$ & perfect\\
any  &  p.r.  & p.r. & $\Bim$\\\hline
\end{tabular}
\renewcommand{\arraystretch}{1}
\label{table:extremepolar}
\end{table}
\vspace*{3mm}

The following proposition states more precisely the way the $\beta$-parameters evolve when one of the input BDEs is extreme.
\begin{proposition}\label{prop:polarextreme}
If $\Bim$ is extreme, then the $\beta$-parameters at the output are given by
\begin{align}\label{eq:polarextreme1}
\Bm & = \begin{cases} 
           \Bip, & \text{if $\Bim$ is perfect}\\
           \frac12, & \text{if $\Bim$ is p.r.};
\end{cases} 
\end{align}
\begin{align}\label{eq:polarextreme2}
\Bp & = \begin{cases} 
0, & \text{if $\Bim$ is perfect}\\
\Bip, & \text{if $\Bim$ is p.r.}.
\end{cases} 
\end{align}
If $\Bip$ is extreme, then \eqref{eq:polarextreme1} and \eqref{eq:polarextreme2} hold after 
interchanging $\Bim$ and $\Bip$.
\end{proposition}

\begin{proof}
Suppose $\Bim\equiv 0$ (perfect), then $\Aim$ can only take the values $0$ and $1$,
and we obtain from \eqref{eq:rel1} that
$$
\Am = \Aim\ast \Aip = \begin{cases} \Aip, & \Aim = 0;\\
\cAip, & \Aim = 1.
\end{cases}
$$
Thus, $\Bm = \min(\Am,\cAm) = \min(\Aip,\cAip) = \Bip$, completing the proof of the first case in \eqref{eq:polarextreme1}.
We skip the proof of the remaining three cases since they follow by similar reasoning.
\end{proof}

\section{Covariance review}\label{Sect:correlation}
In this part, we collect some basic facts about the covariance function, which we will need in the following sections.
The first result is the following formula for splitting a covariance into two parts.
\begin{lemma}\label{lemmacovariancedecomposition}
Let $\bS$, $\bT$ be jointly distributed random vectors over $\R^m$ and $\R^n$, respectively.
Let $f,g:\R^{m+n}\to \R$ be functions such that $\cov[f(\bS,\bT),g(\bS,\bT)]$ exists, {\sl i.e.\/}, $\E f(\bS,\bT)g(\bS,\bT)$, $\E f(\bS,\bT)$, and $\E g(\bS,\bT)$ all exist.
Then,
\begin{align}\label{covariancedecomposition}
\cov[f(\bS,\bT),& g(\bS,\bT)]  = \E_{\bT} \cov_{\bS|\bT}[f(\bS,\bT),g(\bS,\bT)] \nonumber \\
& + \cov_{\bT}[\E_{\bS|\bT}f(\bS,\bT),\E_{\bS|\bT} g(\bS,\bT)].
\end{align}
\end{lemma}
Although this is an elementary result, we give a proof here mainly for illustrating the notation.
Our proof follows \cite{Esary1967}.
\begin{proof}
We will omit the arguments of the functions for brevity.
\begin{align*}
\cov(f,g) & = \E_{\bS,\bT} fg - \E_{\bS,\bT} f \cdot \E_{\bS,\bT} g\\
& = \E_{\bT} \E_{\bS|\bT} fg - \E_{\bT} \big[ \E_{\bS|\bT} f \cdot \E_{\bS|\bT} g\big] \\
&+ \E_{\bT} \big[ \E_{\bS|\bT} f \cdot \E_{\bS|\bT} g\big] - \E_{\bT}\E_{\bS|\bT} f \cdot \E_{\bT}\E_{\bS|\bT} g\\
& = \E_{\bT} \cov_{\bS|\bT}(f , g) + \cov_{\bT} (\E_{\bS|\bT} f, \E_{\bS|\bT} g).
\end{align*}
\end{proof}

The second result we recall is the following inequality.

\begin{lemma}[Chebyshev's covariance inequality]\label{lemma:chebyshev}
Let $X$ be a random variable taking values over $\R$ and let $f,g:\R\to \R$ be any two nondecreasing functions.
Suppose that $\cov(f(X),g(X))$ exists, {\sl i.e.}, $\E f(X)g(X)$, $\E f(X)$, and $\E g(X)$ all exist. Then,
\begin{equation}\label{eq:chebyshev}
\cov(f(X),g(X)) \ge 0.
\end{equation}
\end{lemma}

\begin{proof}
Let $X'$ be an independent copy of $X$.
Let $\E$ and $\E'$ denote expectation with respect to $X$ and $X'$, respectively.
The proof follows readily from the following identity whose proof can be found in \cite[p.~43]{Hardy}. 
\begin{align*}
\cov&(f(X),g(X)) = \E f(X)g(X) - \E f(X) \E g(X)\\
& = \frac{1}{2} \E' \E [(f(X)- f(X'))(g(X)-g(X'))].
\end{align*}
Now note that for any $x,x'\in \R$, $f(x)-f(x')$ and $g(x)-g(x')$ have the same sign since both $f$ and $g$ are
nondecreasing. Thus, $(f(x)- f(x'))(g(x)-g(x'))\ge 0$, and non-negativity of the covariance follows.
\end{proof}

\section{Proof of Theorem~\ref{theoremvarentropy2}}\label{sectionproof}

Let us recall the setting of Theorem~\ref{theoremvarentropy2}.
We have two independent BDEs $(X_1,Y_1)$ and $(X_2,Y_2)$ as inputs of a polar transform,
and two BDEs $(U_1,\bY)$ and $(U_2; U_1,\bY)$ at the output, with $U_1 = X_1\oplus X_2$, $U_2= X_2$, and $\bY =(Y_1,Y_2)$.
Associated with these BDEs are the conditional entropy random variables $\him$, $\hip$, $\hm$, and $\hp$, as defined by \eqref{inputentropies} and \eqref{outputentropies}.
We will carry out the proof mostly in terms of the canonical parameters $A_i\defn \alpha_{X_i|Y_i}(Y_i)$ and $B_i \defn \beta_{X_i|Y_i}(Y_i)$, $i=1,2$.
For shorthand, we will often write $\bX=(X_1,X_2)$, $\bU=(U_1,U_2)$, $\bA=(A_1,A_2)$, and $\bB=(B_1,B_2)$.

We will carry out our calculations in the probability space defined by the joint ensemble $(\bX,\bY)$.
Probabilities over this ensemble will be denoted by $P(\cdot)$ and expectations by $\E[\cdot]$.
Partial and conditional expectations and covariances will be denoted by $\E_{\bY}$, $\E_{\bX|\bY}$, $\cov_{\bY}$, $\cov_{\bX|\bY}$, etc.
Due to the 1-1 nature of the correspondence between $\bU$ and $\bX$, 
expectation and covariance operators such as $\E_{\bU|\bY}$ and $\cov_{\bU|\bY}$ 
will be equivalent to $\E_{\bX|\bY}$ and $\cov_{\bX|\bY}$, respectively.
We will prefer to use expectation operators in terms of the primary variables $\bX$ and $\bY$ rather than the secondary (derived) variables such as $\bU$,
$\bA$, $\bB$, to emphasize that the underlying space is $(\bX,\bY)$.
We note that, due to the independence of $Y_1$ and $Y_2$, $A_1$ and $A_2$ are independent; likewise, $B_1$ and $B_2$ are independent.

\subsection{Covariance decomposition step}
As the first step of the proof of Theorem~\ref{theoremvarentropy2}, we use the covariance decomposition formula \eqref{covariancedecomposition} to write
\begin{align}\label{covar}
\cov(\hm,\hp)  =  & \E_{\bY}\cov_{\bX|\bY}(\hm,\hp ) \nonumber \\
 & + \cov_{\bY}(\E_{\bX|\bY}\hm,\E_{\bX|\bY}\hp ).
\end{align}
For brevity, we will use the notation
$$
\cov_1 \defn \E_{\bY}\cov_{\bX|\bY}(\hm,\hp )
$$
$$
\cov_2 \defn \cov_{\bY}(\E_{\bX|\bY}\hm,\E_{\bX|\bY} \hp)
$$
to denote the two terms on the right hand side of \eqref{covar}.
Our proof of Theorem~\ref{theoremvarentropy2} will consist in proving the following two statements.

\begin{proposition}\label{proposition1}
We have $\cov_1 \ge 0$, with equality iff either $(X_1,Y_1)$ or $(X_2,Y_2)$ is an erasing BDE.
\end{proposition}

\begin{proposition}\label{proposition2}
We have $\cov_2 \ge 0$.
\end{proposition}

\begin{remark}
We note that $\cov_2=0$ iff, of the two BDEs $(X_1,Y_1)$ and $(X_2,Y_2)$, either one is extreme or both are pure.
We note this only for completeness but do not use it in the paper.
\end{remark}

The rest of the section is devoted to the proof of the above propositions.

\subsection{Proof of Proposition~\ref{proposition1}}\label{Sect:proofofproposition1}
For $p,q \in [0,1]$, define  
\begin{align}\label{covarianceformula1}
f(p,q) & \defn (p \ast q)(p\ast \overline{q}) 
\log \left(\frac{p \ast q}{p \ast \overline{q}}\right) \nonumber\\
& \qquad\times
\left[\cH\left(\frac{p\, \overline{q}}{p \ast \overline{q}}\right) -
\cH\left(\frac{p\, q}{p \ast q}\right) \right].
\end{align}
We will soon give a formula for $\cov_1$ in terms of this function.
First, a number of properties of $f(p,q)$ will be listed.
The following symmetry properties are immediate:
\begin{equation}\label{symmetry1}
f(p,q) = f(\overline{p},q)=f(\overline{p},q) = f(\overline{p},\overline{q}),
\end{equation}
\begin{equation}\label{symmetry2}
f(p,q) = f(q,p).
\end{equation}

\begin{lemma}\label{Cpq}
We have $f(p,q)\ge 0$ for all $p,q\in [0,1]$ with equality iff $p \in \{0,1/2,1\}$ or $q\in \{0,1/2,1\}$.
\end{lemma}
\begin{proof}
We use \eqref{symmetry1} to write
\begin{equation}\label{symmetry3}
f(p,q) = f(r,s)
\end{equation}
where $r\defn \min\{p,\overline{p}\}$ and $s\defn \min\{q,\overline{q}\}$.
Thus, instead of proving $f(p,q)\ge 0$, it suffices to prove $f(r,s)\ge 0$ for $0\le r,s\le 1/2$.
In fact, using \eqref{symmetry2}, it suffices to prove $f(r,s)\ge 0$ for $0\le r\le s \le 1/2$.
Assuming $0\le r\le s \le 1/2$, it is straightforward to show that
\begin{equation}\label{conditionsrs}
r\ast s \ge r\ast \overline{s}  \quad \text{and}\quad
\frac{rs}{r\ast s} \,\le\, 
\frac{r\,\overline{s}}{r\ast \overline{s}} \, \le \,
\frac12.
\end{equation}
Thus, if we write out the expression for $f(r,s)$, as in \eqref{covarianceformula1} with $(r,s)$ in place of $(p,q)$, we can see easily that each of the four factors on the right hand side of that expression are non-negative. More specifically, the logarithmic term is non-negative due to the first inequality in \eqref{conditionsrs} and the bracketed term is non-negative due to the second inequality in \eqref{conditionsrs}.
This completes the proof that $f(p,q)\ge 0$ for all $p,q\in [0,1]$.

Next, we identify the necessary and sufficient conditions for $f(p,q)$ to be zero over $0\le p,q\le 1$.
Clearly, $f(p,q) = 0$ iff one of the four factors on the right hand side of \eqref{covarianceformula1} equals zero.
By straightforward algebra, one can verify the following statements. 
The first factor $p\ast q$ equals zero iff $(p,q)\in \{(0,1),(1,0)\}$.
The second factor $p\ast \overline{q}$ equals zero iff $(p,q)\in \{(0,0),(1,1)\}$.
The log term equals zero iff $p = 1/2$ or $q=1/2$.
Finally the difference of the entropy terms equals zero iff $pq/p\ast q = p\overline{q}/p\ast \overline{q}$ or $pq/p\ast q =1- p\overline{q}/p\ast \overline{q}$ which in turn is true iff $p \in \{0,1/2,1\}$ or $q \in \{0,1/2,1\}$.
Taking the logical combination of these conditions we conclude that $f(p,q) = 0$ iff $p \in \{0,1/2,1\}$ or $q \in \{0,1/2,1\}$. 
\end{proof}

\begin{lemma}\label{cov1formula}
We have  
\begin{equation}\label{formula1}
\cov_1 = \E f(\bA) = \E f(\bB).
\end{equation}
\end{lemma}

\begin{proof}
Fix a sample $\by=(y_1,y_2)$.
Note that
\begin{align*}
&\cov_{\bX|\by}(\hm,\hp) = \cov_{\bX|\by}(h(U_1|\by),h(U_2|U_1,\by))\\
&\quad = \E_{\bX|\by}\big\{\big[h(U_1|\by)-H(U_1|\by)\big]h(U_2|U_1,\by)\big\}\\
&\quad = \sum_{u_1} p_{U_1|\bY}(u_1|\by) \big[h(u_1|\by)-H(U_1|\by)\big] H(U_2|u_1,\by).
\end{align*}
After some algebra, the term $\big[h(u_1|\by)- H(U_1|\by)\big]$ simplifies to 
$$
(1-p_{U_1|\bY}(u_1|\by))\log \frac{1-p_{U_1|\bY}(u_1|\by)}{p_{U_1|\bY}(u_1|\by)}.
$$ 
Substituting this in the preceding equation and writing out the sum over $U_1$ explicitly, we obtain
\begin{align*}
&\cov_{\bX|\by}(\hm,\hp) = p_{U_1|\bY}(0|\by) p_{U_1|\bY}(1|\by)  \\
&\quad \cdot\log \frac{p_{U_1|\bY}(0|\by)}{p_{U_1|\bY}(1|\by)} \big[H(U_2|U_1=1,\by)-H(U_2|U_1=0,\by)\big].
\end{align*}
Expressing each factor on the right side of the above equation in terms of $a_i = \alpha(y_i)$, $i=1,2$, we see
that it equals $f(a_1,a_2)$.
Taking expectations, we obtain $\cov_1 = \E f(\bA)$.
The alternative formula $\cov_1 = \E f(\bB)$ follows from the fact that $f(\bB) = f(\bA)$ due to the symmetries \eqref{symmetry1}.
\end{proof}

Proposition~\ref{proposition1} now follows readily.
We have $\cov_1\ge 0$ since $f(a_1,a_2)\ge 0$ for all $a_1, a_2\in [0,1]$ by Lemma~\ref{Cpq}.
By the same lemma, strict positivity, $\E f(\bA)>0$, is possible iff the events 
$A_1\notin \{0,1/2,1\}$ and $A_2\notin \{0,1/2,1\}$ can occur simultaneously with non-zero probability,
{\sl i.e.\/}, iff
\begin{equation}\label{condition1}
P\bigg(A_1\notin \{0,\frac12,1\}\bigg)\;P\bigg(A_2\notin \{0,\frac12,1\}\bigg) > 0,
\end{equation}
since $A_1$ and $A_2$ are independent. 
Condition~\eqref{condition1} is true iff  
\begin{equation}\label{condition2}
P\bigg(B_1\notin \{0,\frac12\}\bigg)\; P\bigg(B_2\notin \{0,\frac12\}\bigg) > 0,
\end{equation}
which in turn is true iff neither $B_1$ nor $B_2$ is erasing.
This completes the proof of Proposition~\ref{proposition1}.

\subsection{Proof of Proposition~\ref{proposition2}}\label{Sect:proofofproposition2}
Let $\Hm(p,q)  \defn \cH(p \ast q)$ and $\Hp(p,q)  \defn \cH(p) + \cH(q) - \cH(p \ast q)$ for $p,q\in [0,1]$.
These functions will be used to give an explicit expression for $\cov_2$.
First, we note some symmetry properties of the two functions.
For $i=1,2$, we have
\begin{equation}\label{symmetry4}
g_i(p,q) = g_i(\overline{p},q)=g_i(\overline{p},q) = g_i(\overline{p},\overline{q}),
\end{equation}
\begin{equation}\label{symmetry5}
g_i(p,q) = g_i(q,p).
\end{equation}
We omit the proofs since they are immediate.

\begin{lemma}\label{cov2formula}
We have, for $i=1,2$,
\begin{equation}\label{cov2forms12}
\E_{\bX|\bY} \hoi = g_i(\bA) = g_i(\bB).
\end{equation}
\end{lemma}

\begin{proof}
These results follow from \eqref{eq:entropyA}, \eqref{eq:rel1}, and \eqref{eq:rel2}.
We compute $\E_{\bX|\bY} \hm$ as follows.
\begin{align*}
\E_{\bX|\bY} \hm  &= \E_{\bU|\bA} \hm = \cH(A_1\ast  A_2) = \Hm(\bA).
\end{align*}
For the second term, we use the entropy conservation \eqref{conservation1}.
\begin{align*}
\E_{\bX|\bY} \hp &= \E_{\bX|\bY} \him + \E_{\bX|\bY} \hip - \E_{\bX|\bY} \hm \\
& = \cH(A_1) + \cH(A_2) - \cH(A_1\ast  A_2) = \Hp(\bA). 
\end{align*}
The second form of the formulas in terms of $\bB$ follow from the symmetry properties \eqref{symmetry4}.
\end{proof}

As a corollary to Lemma~\ref{cov2formula}, we now have
\begin{equation}\label{cov2form2}
\cov_2 = \cov[\Hm(\bB),\Hp(\bB)].
\end{equation}

In order to prove that $\cov_2\ge 0$, we will apply Lemma~\ref{lemma:chebyshev} to \eqref{cov2form2}.
First, we need to establish some monotonicity properties of the functions $\Hm$ and $\Hp$.
We insert here a general definition.
\begin{definition}\label{def:increasing}
A function $g:\R^n \to \R$ is called {\sl nondecreasing} if, for all 
$\bx,\by\in \R^n$, $g(\bx)\le g(\by)$ whenever $x_i\le y_i$ for all $i=1,\ldots,n$.
\end{definition}

\begin{lemma}\label{cov2monotonicity}
$\Hm:[0,1/2]^2 \to \R^+$ is nondecreasing.
\end{lemma}

\begin{proof}
Since $\Hm(b_1,b_2) = \Hm(b_2,b_1)$, it suffices to show that 
$\Hm(b_1,b_2)$ is nondecreasing in $b_1\in[0,1/2]$ for fixed $b_2\in [0,1/2]$.
So, fix $b_2\in [0,1/2]$ and consider $\Hm(b_1,b_2)$ as a function of $b_1\in [0,1/2]$.
Recall the well-known facts that the function $\cH(p)$ over $p\in [0,1]$ is a strictly concave non-negative function, symmetric around $p=1/2$, attaining its minimum value of 0 at $p\in \{0,1\}$, and its maximum value of $1$ at $p=1/2$. 
It is readily verified that, for any fixed $b_2\in [0,1/2]$, as $b_1$ ranges from 0 to $1/2$, $b_1\ast b_2$ decreases from $\overline{b}_2$ to $1/2$,
hence $\Hm(b_1,b_2) = \cH(b_1\ast b_2)$ increases from $\cH(\overline{b}_2)$ to $\cH(1/2) = 1$, with strict monotonicity if $b_2\neq 1/2$.
This completes the proof.
\end{proof}

\begin{lemma}\label{cov2monotonicity2}
$\Hp:[0,1/2]^2 \to \R^+$ is nondecreasing.
\end{lemma}

\begin{proof}
Again, since $\Hp(b_1,b_2) = \Hp(b_2,b_1)$, it suffices to show that 
$\Hp(b_1,b_2)$ is nondecreasing in $b_1\in[0,1/2]$ for fixed $b_2\in [0,1/2]$.
Recall that $\Hp(b_1,b_2)=\cH(b_1)+\cH(b_2)-\cH(b_1\ast b_2)$.
Exclude the constant term $\cH(b_2)$ and focus on the behavior of $I(b_1) \defn \cH(b_1\ast b_2)-\cH(b_1)$ over $b_1\in[0,1/2]$. 
Observe that $I(b_1)$ is the mutual information between the input and output terminals of a BSC with crossover probability $b_1$ and a Bernoulli-$b_2$ input.
The mutual information between the input and output of a discrete memoryless channel is a convex function of the set of channel transition probabilities for any fixed input probability assignment \cite[p.~90]{GallagerBook}. So, $I(b_1)$ is convex in $b_1\in[0,1/2]$. 
Since $I(0)=\cH(b_2)$ and $I(1/2)=0$, it follows from the convexity property that $I(b_1)$ is decreasing in $b_1\in [0,1/2]$, and strictly decreasing if $b_2\neq 0$.
This completes the proof.
\end{proof}

Proposition~\ref{proposition2} can now be proved as follows.
First, we apply Lemma~\ref{lemmacovariancedecomposition} to \eqref{cov2form2} to decompose $\cov_2$ as 
\begin{align*}
\cov(\Hm(\bB),\Hp(\bB)) & = \E_{B_1} \cov_{B_2} (\Hm(\bB),\Hp(\bB))\\
& + \cov_{B_1} (\E_{B_2}\Hm(\bB),\E_{B_2}\Hp(\bB)).
\end{align*}
Each covariance term on the right side is positive by Chebyshev's correlation inequality (Lemma~\ref{lemma:chebyshev}) 
and the fact that $\Hm$ and $\Hp$ are nondecreasing in the sense of Def.~\ref{def:increasing}.
More specifically, Chebyshev's inequality implies that
$$
\cov_{B_2} (\Hm(b_1,B_2),\Hp(b_1,B_2))\ge 0
$$
for any fixed $b_1\in [0,1/2]$ since $\Hm(b_1,b_2)$ and $\Hp(b_1,b_2)$ are nondecreasing functions of $b_2$ when $b_1$ is fixed.
Likewise, Chebyshev's inequality implies that
$$
\cov_{B_1} (\E_{B_2}\Hm(\bB),\E_{B_2}\Hp(\bB)) \ge 0
$$
since $\E_{B_2}\Hm(b_1,B_2)$ and $\E_{B_2}\Hp(b_1,B_2)$ are, as a simple consequence of Lemma~\ref{cov2monotonicity2}, nondecreasing functions of $b_1$.

\subsection{Proof of Theorem~\ref{theoremvarentropy2}}

The covariance inequality \eqref{covarianceinequality} is an immediate consequence of 
\eqref{covar} and Propositions~\ref{proposition1} and \ref{proposition2}.
We only need to identify the necessary and sufficient conditions for the covariance to be zero.
For brevity, let us define 
$$
T\defn \text{``$B_1$ or $B_2$ is extreme''}.
$$
The present goal is to prove that
\begin{equation}\label{eq:goal}
\text{$\cov(\hm,\hp)=0\quad$ iff $\quad T$ holds.}
\end{equation}
The proof will make us of the decomposition
\begin{align}\label{eq:covB}
\cov(\hm,\hp)& =\cov_1 + \cov_2\nonumber\\
& = \E f(\bB) + \cov(\Hm(\bB),\Hp(\bB))
\end{align}
that we have already established.
Let us define 
$$
R\defn \text{``$B_1$ or $B_2$ is erasing''}
$$ 
and note that $R$ appears in Proposition~\ref{proposition1} as the necessary and sufficient conditions for $\cov_1$ to be zero.
Note also that $T$ implies $R$ since ``extreme'' is a special instance ``erasing'' according to definitions in Table~\ref{table:classification}.

We begin the proof of \eqref{eq:goal} with the sufficiency part.
in other words, by assuming that $T$ holds. 
Since $T$ implies $R$, $T$ is sufficient for $\cov_1=0$. 
To show that $T$ is sufficient for $\cov_2=0$, we recall Proposition~\ref{prop:polarextreme},
which states that, if $T$ is true, then either $\Bm$ or $\Bp$ is extreme.
To be more specific, if $\Bim$ or $\Bip$ is p.r., then $\Bm\equiv 1/2$ and $\Hm(\bB)\equiv 1$;
if $\Bim$ or $\Bip$ is perfect, then $\Bp\equiv 0$ and $\Hp(\bB)\equiv 0$. (The notation ``$\equiv$'' should be read as ``equals with probability one''.) 
In either case, $\cov_2=\cov(\Hm(\bB),\Hp(\bB))=0$.
This completes the proof of the sufficiency part.

To prove necessity in \eqref{eq:goal}, 
we write $T$ as
\begin{equation}\label{eq:T}
T = R\wedge (R^c\vee T)
\end{equation}
where $R^c$ denotes the complement (negation) of $R$.
The validity of \eqref{eq:T} follows from $R\wedge T = T$.
To prove necessity, we will use contraposition
and show that $T^c$ implies $\cov(\hm,\hp)>0$.
Note that $T^c = R^c \vee (R\wedge T^c)$.
If $T^c$ is true, then either $R^c$ or $(R\wedge T^c)$ is true.
If $R^c$ is true, then $\cov_1 >0$ by Proposition~\ref{proposition1}. 
We will complete the proof by showing that $R\wedge T^c$ implies $\cov(\hm,\hp)>0$.
For this, we note that when one of the BDEs is erasing, there is an explicit formula
for $\cov_2$. We state this result as follows.

\begin{lemma}\label{lemma:BECcov2}
Let $B_1$ be erasing with erasure probability $\epsilon \defn P(B_1 = 1/2)$ and let $B_2$ be arbitrary with $\delta \defn H(X_2|Y_2)$. 
Then, 
\begin{equation}\label{eq:cov2BEC}
\cov_2 = \epsilon(1-\epsilon)\delta(1-\delta)
\end{equation}
This formula remains valid if $B_2$ is erasing with erasure probability $\epsilon \defn P(B_2 = 1/2)$ and $B_1$ is arbitrary with $\delta \defn H(X_1|Y_1)$.
\end{lemma}
\begin{proof}
We first observe that
$$
\Hm(B_1,B_2) = \begin{cases} \cH(B_2), &  B_1 =0;\\
1, & B_1 = \frac12;
\end{cases}
$$
$$
\Hp(B_1,B_2) =  \begin{cases} 0, &  B_1 =0;\\
\cH(B_2), & B_1 = \frac12.
\end{cases}
$$
Now, the claim \eqref{eq:cov2BEC} is obtained by simply computing the covariance of these two random variables.
The second claim follows by the symmetry property \eqref{symmetry5}.
\end{proof}

Returning to the proof of Theorem~\ref{theoremvarentropy2}, the proof of the necessity part is now completed as follows.
If $R\wedge T^c$ holds, then at least one of the BDEs is {\sl strictly} erasing (has erasure probability $0 <\epsilon <1$) and the other is non-extreme.
By Proposition~\ref{prop:extreme}, the conditional entropy $H(X|Y)$ of a non-extreme BDE $(X,Y)$ is strictly between 0 and 1.
So, by Lemma~\ref{lemma:BECcov2}, we have $\cov_2>0$. 
This completes the proof.

\section{Varentropy under higher-order transforms}\label{higherorder}

In this part, we consider the behavior of varentropy under higher-order polar transforms.
The section concludes with a proof of the polarization theorem using properties of varentropy.
 
\subsection{Polar transform of higher orders}
For any $n\ge 1$, there is a polar transform of order $N=2^n$.
A polar transform of order $N=2^n$ is a mapping $\psi_N$ that takes $N$ BDEs $\{(X_i,Y_i)\}_{i=1}^N$, as input, and produces a new set of $N$
BDEs $\{(U_i; \bU^{i-1},\bY)\}_{i=1}^N$, where $\bY=(Y_1,\ldots,Y_N)$ and $\bU^{i-1}=(U_1,\ldots,U_{i-1})$ is a subvector of $\bU=(U_1,\ldots,U_N)$, which in turn
is obtained from $\bX=(X_1,\ldots,X_N)$ by the transform
\begin{equation}\label{eq:polarN}
\bU =\bX \bG_N, \qquad \bG_N \defn \bF^{\otimes n}, \qquad \bF\defn \begin{bmatrix}1 & 0 \\ 1 & 1 \end{bmatrix}.
\end{equation}
The sign ``$\otimes n$'' in the exponent denotes the $n$th Kronecker power.
We allow $Y_i$ to take values in some arbitrary set $\cY_i$, $1\le i\le N$, which is not necessarily discrete.
We assume that $(X_i,Y_i)$, $1\le i\le N$, are independent but not necessarily identically-distributed.

(An alternate form of the polar transform matrix, as used in \cite{ArikanIT2009}, is $\bG_N = \bB_N \bF^{\otimes n}$, in which $\bB_N$ is a permutation matrix known as {\sl bit-reversal}. The form of $\bG_N$ that we are using here is less complex and adequate for the purposes of this paper. 
However, if desired, the results given below can be proved under bit-reversal (or, any other permutation) after suitable re-indexing of variables.)

\subsection{Polarization results}

The first result in this section is a generalization of Theorem~\ref{varentropycontraction0} to higher order polar transforms.

\begin{theorem}\label{varentropygeneral}
Let $N=2^n$ for some $n\ge 1$. Let $(X_i,Y_i)$, $1\le i\le N$, be independent but not necessarily identically distributed BDEs.
Consider the polar transform $\bU= \bX\bG_N$ and let $(U_i;\bU^{i-1},\bY)$, $1\le i\le N$, be the BDEs at the output of the polar transform.
The varentropy is nonincreasing under any such polar transform in the sense that
\begin{equation}\label{varentropyinequality}
\sum_{i=1}^N V(U_i|\bU^{i-1},\bY) \le \sum_{i=1}^N V(X_i|Y_i).
\end{equation}
\end{theorem}

The next result considers the special case in which the BDEs at the input of the polar transform are i.i.d. and the transform size goes to infinity.

\begin{theorem}\label{varentropyasymptotic}
Let $(X_i,Y_i)$, $1\le i\le N$, be i.i.d. copies of a given BDE $(X,Y)$.
Consider the polar transform $ \bU = \bX\bG_N$ and let $(U_i;\bU^{i-1},\bY)$, $1\le i\le N$, be the BDEs at the output of the polar transform.
Then, the average varentropy at the output goes to zero asymptotically:
\begin{equation}\label{varentropyconvergence}
\frac1N \sum_{i=1}^N V(U_i|\bU^{i-1},\bY) \to 0, \quad \text{as $N \to \infty$.}
\end{equation}
\end{theorem}

\subsection{Proof of Theorem~\ref{varentropygeneral}}

We will first bring out the recursive nature of the polar transform by giving a more abstract formulation in terms of the $\alpha$-parameters of the
variables involved. 
Let us recall that a polar transform of order two is essentially a mapping of the form 
\begin{equation}\label{A2A}
(A_{\text{in},1},A_{\text{in},2})\to (A_{\text{out},1},A_{\text{out},2}),
\end{equation}
where $A_{\text{in},1}$ and $A_{\text{in},2}$ are the $\alpha$-parameters of the input BDEs $(X_1,Y_1)$ and $(X_2,Y_2)$,
and $A_{\text{out},1}$ and $A_{\text{out},2}$ are the $\alpha$-parameters of the output BDEs $(U_1,\bY)$ and $(U_2;U_1,\bY)$.

Alternatively, the polar transform may be viewed as an operation in the space of CDFs of $\alpha$-parameters
and represented in the form
\begin{equation}\label{F2F}
(F_{\text{out},1},F_{\text{out},2}) = \psi_2(F_{\text{in},1},F_{\text{in},2})
\end{equation}
where $F_{\text{in},i}$ and $F_{\text{out},i}$ are the CDFs of $A_{\text{in},i}$ and $A_{\text{out},i}$, respectively. 

Let $\cM$ be the space of all CDFs belonging to random variables defined on the interval $[0,1]$.
The CDF of any $\alpha$-parameter $A$ belongs to $\cM$, and conversely, each CDF $F\in \cM$ defines a valid $\alpha$-parameter $A$.
Thus, we may regard the polar transform of order two \eqref{F2F} as an operator of the form 
\begin{equation}\label{M2M}
\psi_2 : \cM^2 \to \cM^2.
\end{equation}
We will define higher order polar transforms following this viewpoint.

For each $i=1,\ldots,N$, let $A_{\text{in},i}$ denote the $\alpha$-parameter of the $i$th BDE $(X_i,Y_i)$ at the input, and let  
$F_{\text{in},i}$ denote the CDF of $A_{\text{in},i}$.
Likewise, let $A_{\text{out},i}$ denote the $\alpha$-parameter of the $i$th BDE $(U_i;\bU^{i-1},\bY)$
at the output, and let $F_{\text{out},i}$ be the CDF of $A_{\text{out},i}$.
Let $\bFi=(F_{\text{in},1},\ldots,F_{\text{in},N})$ and $\bFo=(F_{\text{out},1},\ldots,F_{\text{out},N})$.
We will represent a polar transform of order $N$ abstractly as $\bFo = \psi_N(\bFi)$.

There is a recursive formula that defines the polar transform of order $N$ in terms of the polar transform of order $N/2$.
Let us split the output $\bFo$ into two halves as $\bFo = (\bFo',\bFo'')$.
Each half is obtained by a size-$N/2$ transform of the form
\begin{equation*}
\bFo' = \psi_{N/2}(\bFi'), \quad \bFo'' = \psi_{N/2}(\bFi''),
\end{equation*}
in which $\bFi' =(F_{\text{in},1}',\ldots,F_{\text{in},N/2}')$, $\bFi'' =(F_{\text{in},1}'',\ldots,F_{\text{in},N/2}'')$ are obtained from $\bFi$ through
a series of size-2 transforms
\begin{equation}\label{recursionF}
(F_{\text{in},i}',F_{\text{in},i}'') = \psi_2(F_{\text{in},i},F_{\text{in},i+N/2}), \quad 1\le i \le N/2.
\end{equation}
The derivation of the above recursion from the algebraic definition \eqref{eq:polarN} is standard knowledge in polar coding, and will be omitted.

Let us write $V(F)$ to denote the varentropy $V(X|Y)$ of a BDE $(X,Y)$ whose $\alpha$-parameter has CDF $F$.
Using \eqref{eq:varentropyA}, we can write $V(F)$ as 
\begin{equation}
V(F) = \int_0^1 \cH_2(a) \diff F(a) - \biggl(\int_0^1 \cH(a) \diff F(a)\biggr)^2. \label{def:VF}
\end{equation}

We are now ready to prove Theorem~\ref{varentropygeneral}.
The proof will be by induction. First note that the claim \eqref{varentropyinequality} is true for $N=2$ by Theorem~\ref{theoremvarentropy}.
Let $N\ge 4$ and suppose, as induction hypothesis, that the claim is true for transforms of orders $N/2$ and smaller. We will show that the claim is true for order $N$.
By the induction hypothesis, we have 
\begin{equation}\label{eq:ineq1}
\sum_{i=1}^{N/2} V(F_{\text{out},i}') \le \sum_{i=1}^{N/2} V(F_{\text{in},i}')
\end{equation}
and
\begin{equation}\label{eq:ineq2}
\sum_{i=1}^{N/2} V(F_{\text{out},i}'') \le \sum_{i=1}^{N/2} V(F_{\text{in},i}'').
\end{equation}
Summing \eqref{eq:ineq1} and \eqref{eq:ineq2} side by side,
\begin{equation}\label{eq:ineq3}
\sum_{i=1}^{N} V(F_{\text{out},i}) \le \sum_{i=1}^{N/2}  \bigg[V(F_{\text{in},i}')+V(F_{\text{in},i}'')\bigg]
\end{equation}
Using the induction hypothesis again, we obtain
\begin{equation}\label{eq:ineq4}
V(F_{\text{in},i}') + V(F_{\text{in},i}'') \le V(F_{\text{in},i}) + V(F_{\text{in},i+N/2})
\end{equation}
for all $i=1,\ldots,N/2$.
The proof if completed by using \eqref{eq:ineq4} to upper-bound the right side of \eqref{eq:ineq3} further.

\subsection{Proof of Theorem~\ref{varentropyasymptotic}}
In this proof we will consider a sequence of polar transforms indexed by $n\ge 1$.
For a given $n$, the size of the transform is $N=2^n$; the inputs of the transform are $(X_i,Y_i)$, $1\le i\le N$, which are i.i.d. copies of a given BDE $(X,Y)$;
the outputs of the transform, which we will refer to as ``the $n$th generation BDEs'', are $(U_i;\bU^{i-1},\bY)$, $1\le i\le N$. Let $F_0$ denote the CDF of $(X,Y)$.
Let $F_{n,i}$ denote the CDF of $(U_i;\bU^{i-1},\bY)$,
the $i$th BDE in the $n$th generation, $n\ge 1$, $1\le i\le 2^n$, and set $F_{0,1}=F_0$.
In this notation, we can express the normalized varentropy compactly as
$$
\overline{V}_n \defn \frac{1}{2^n} \sum_{i=1}^{2^n} V(U_i|\bU^{i-1},\bY) = \frac{1}{2^n} \sum_{i=1}^{2^n} V(F_{n,i}), \quad n\ge 1,
$$ 
and $\overline{V}_0 \defn V(F_0)$.
The sequence $\{\overline{V}_n\}$ is non-negative (since each $\overline{V}_n$ is a sum of varentropies), and nonincreasing by Theorem~\ref{varentropygeneral}.
Thus $\{\overline{V}_n\}$ converges to a limit $c\ge 0$.
Our goal is to prove that $c = 0$.

The analysis in the proof of Theorem~\ref{varentropygeneral} covers the present case as a special instance.
In the present notation, the recursive relation \eqref{recursionF} takes the form
\begin{equation*}
(F_{n,i},F_{n,i+2^{n-1}}) = \psi_2(F_{n-1,i},F_{n-1,i}),\quad 1\le i\le 2^{n-1},
\end{equation*}
since here we have $F_{n-1,i}=F_{n-1,i+2^{n-1}}$ due to i.i.d. BDEs at the transform input.
Using this relation, we obtain readily an explicit formula for the incremental change in normalized varentropy from generation $n$ to $(n+1)$, namely,
\begin{equation}\label{def:Dn}
D_{n+1} \defn \overline{V}_{n+1} - \overline{V}_{n} = - \sum_{i=1}^{2^{n}} C(F_{n,i}), \quad n\ge 0,
\end{equation}
where
\begin{equation}\label{def:C}
C(F_{n,i}) \defn V(F_{n,i})-\big[V(F_{n+1,i}) + V(F_{n+1,i+2^{n}})\big]/2.
\end{equation}
If we denote the conditional entropy random variables in the polar transform as $\{h_{n,i}\}$,
it can be seen that 
$$
C(F_{n,i}) = \cov(h_{n+1,i},h_{n+1,i+2^n}).
$$
Thus, we have $C(F_{n,i})\ge 0$ by Theorem~\ref{theoremvarentropy2}, implying that $D_{n}\le 0$ for all $n\ge 1$.
It is useful to note here that
\begin{equation}\label{csum}
c \defn \lim_{n\to \infty} \overline{V}_n =  V(F_0) - \sum_{i=1}^\infty D_n,
\end{equation}
showing explicitly that $c$ is the limit of a monotone nonincreasing sequence of sums.

For $\delta \ge 0$, let 
\begin{equation}\label{def:Mdelta}
\cM_{\delta} \defn \{F\in \cM: V(F) \ge \delta\}.
\end{equation}
and
\begin{equation}\label{def:Delta}
\Delta(\delta) \defn \inf\{C(F): F\in \cM_\delta\}.
\end{equation}
As we will see in a moment, the main technical problem that remains is to show that 
\begin{equation}\label{infassumption}
\delta >0 \implies \Delta(\delta) >0.
\end{equation}
While this proposition seems plausible in view of the fact that $C(F)=0$ iff $V(F)=0$ (by Theorem~\ref{theoremvarentropy2}),
there is the technical question of whether the ``inf'' in \eqref{def:Delta} is achieved as a ``min'' by some $F\in \cM_{\delta}$.
We will first complete the proof of Theorem~\ref{varentropyasymptotic} by assuming that \eqref{infassumption} holds.
Then, we will give a proof of \eqref{infassumption} in the Appendix.

Let $J_n(\delta)\defn \{1\le i\le 2^n: F_{n,i}\in \cM_{\delta}\}$, and $P_n(\delta) \defn |J_n(\delta)|/2^n$.
For $\delta >0$, we may think of $J_n(\delta)$ as the set of ``bad'' BDEs in the $n$th generation and $P_n(\delta)$ as their fraction in the same population.
From \eqref{def:Dn}, we obtain the bound
\begin{equation}\label{DnBound}
D_n \le -P_n(\delta) \Delta(\delta), \quad \delta \ge 0.
\end{equation}
To apply this bound effectively, we need a lower bound on $P_n(\delta)$.
To derive such a lower bound, we observe that, for any $\delta \ge 0$,
\begin{equation}\label{eq:VnLB}
\overline{V}_n  \le [1-P_n(\delta)]\delta+ P_n(\delta)M  \le \delta + P_n(\delta)M
\end{equation}
where $M\defn 2.3434$ is the bound on varentropy provided by Lemma~\ref{lemma:bound}.
Let $n_0$ be such that for all $n\ge n_0$, $\overline{V}_n \ge c/2$. 
Since $\{\overline{V}_n\}$ converges to $c\ge 0$, $n_0$ exists and is finite.
This, combined with \eqref{eq:VnLB}, implies the following bound on the fraction of bad indices.
\begin{equation}\label{PnBound}
P_n(\delta) \ge \frac{\overline{V}_n - \delta}{M}\ge \frac{c/2-\delta}{M}, \qquad n\ge n_0.  
\end{equation}
Using \eqref{PnBound} in \eqref{DnBound} with $\delta = c/4$ gives
\begin{equation}\label{DnBound2}
D_n \le - (c/4M)\cdot\Delta(c/4), \quad n\ge n_0.
\end{equation}
From \eqref{csum}, we see that having $c>0$ is incompatible with \eqref{DnBound2}.
This completes the proof that $c=0$ (subject to the assumption that \eqref{infassumption} holds, which is proved in the Appendix).

\section{Concluding remarks}
One of the implications of the convergence of average varentropy to zero is that the entropy random variables 
``concentrate'' around their means along almost all trajectories of the polar transform. 
This concentration phenomenon provides a theoretical basis for understanding why polar decoders are robust against quantization of likelihoood ratios \cite{hassani2012}.

Theorem~\ref{varentropyasymptotic} may be seen as an alternative version of the ``polarization'' results of \cite{ArikanIT2009}.
In \cite{ArikanIT2009}, the analysis was centered around the mutual information function and martingale methods were used to establish asymptotic results.
The present study is centered around the varentropy and uses weak convegence of probability distributions.
The use of weak convergence in such problems is not new; Richardson and Urbanke \cite[pp.~187-188]{MCT} used similar methods to deal with problems of convergence of functionals defined on the space of binary memoryless channels. 

We should mention that Alsan and Telatar \cite{AlsanTelatar} have given an elementary proof of polarization that avoids martingale theory, and instead, uses Mrs. Gerber's lemma \cite{Gerber}.
It appears possible to adopt the method of \cite{AlsanTelatar} to establish Theorem~\ref{varentropyasymptotic} without using weak convergence.

\appendix[Proof of \eqref{infassumption}]

\begin{lemma}\label{lemma:Compact}
The space $\cM$ of CDFs on $[0,1]$ is a compact metric space.
\end{lemma}

\begin{proof}
This follows from a general result about probability measures on compact metric spaces.
Theorem~6.4 in \cite[p.~45]{Parthasarathy} states that, for any compact metric space $X$, the space $\cM(X)$ of all probability measures defined on the $\sigma$-algebra of Borel sets in $X$ is compact.
Our definition of $\cM$ above coincides with the $\cM(X)$ with $X=[0,1]$.  
\end{proof}

For $F\in \cM$, let $F^-$ and $F^+$ be defined by (see \eqref{M2M})
$$
(F^-,F^+) =\psi_2(F,F).
$$
Define $C:\cM \to \R$ as the mapping 
\begin{equation}\label{def:Cbis}
C(F) \defn V(F)-\big[V(F^-) + V(F^+)\big]/2.
\end{equation}
This definition is a repetition of \eqref{def:C} in a more convenient notation.
We have already seen the interpretation of $C(F)$ as a covariance and mentioned that $C(F)\ge 0$.
It is also clear that $C(F)$ is bounded: $C(F)\le V(F)\le M$, where $M=2.3434$.
Thus, we may restrict the range of $C$ and write it as a mapping $C:\cM\to [0,M]$.

\begin{lemma}\label{lemma:Continuous}
The mapping $C:\cM \to [0,M]$ is continuous (w.r.t. the weak topology on $\cM$ and the usual topology of Borel sets in $\R$).
\end{lemma}

\begin{proof}
We wish to show that if $F_n\Rightarrow F_0$ (in the sense of weak-convergence), then $|C(F_n)-C(F_0)|\to 0$.
We observe from \eqref{def:VF} that $V(F)$ is given in terms of expectations of two bounded uniformly continuous functions, $\cH:[0,1]\to [0,1]$ and $\cH_2:[0,1]\to [0,M]$.
Thus, by definition of weak convergence (\cite[p.~40]{Parthasarathy}), we have $|V(F_n)-V(F_0)|\to 0$.
In view of \eqref{def:Cbis}, the proof will be complete if we can show that $(F_n\Rightarrow F_0)$ implies $(F_n^-\Rightarrow F_0^-)$ and $(F_n^+\Rightarrow F_0^+)$,
where $F_n^- \defn (F_n)^-$, etc.
By the ``portmanteau'' theorem (see, {\sl e.g.}, Theorem~6.1 in \cite[p.~40]{Parthasarathy}), it is sufficient to show that
for every open set $G\subset [0,1]$,
\begin{gather}
\liminf_n \int_{G} \diff F_n^-(a) \ge \int_{G} \diff F_0^-(a),\label{port1}\\
\liminf_n \int_{G} \diff F_n^+(a) \ge \int_{G} \diff F_0^+(a).\label{port2}
\end{gather}

To prove \eqref{port1}, let $f_1:[0,1]^2 \to [0,1]$ be such that $f_1(a_1,a_2) = a_1\ast a_2$.
Then, we can write 
\begin{equation*}
P_n^-(G) \defn \int_{G} \diff F_n^-(a) = \iint\limits_{f_1^{-1}(G)} \diff F_n(a_1) \diff F_n(a_2),
\end{equation*}
which follows from the density evolution equation
\begin{align*}
F_n^-(a) & = \iint\limits_{a_1\ast a_2 \le a} \diff F_n(a_1) \diff F_n(a_2)
\end{align*}
that was proved as part of Proposition~\ref{prop:densityevolution}.
We note that (i) the pre-image $f_1(G)\subset [0,1]^2$ is an open set since the function $f$ is a continuous
and (ii) the product measure $F_n \times F_n$ converges weakly to $F_0\times F_0$ \cite[p.~21, Thm.~3.2]{Billingsley}; so, again by the portmanteau theorem,
\begin{equation*}
\liminf_n \iint\limits_{f_1^{-1}(G)} \diff F_n(a_1) \diff F_n(a_2) \ge \iint\limits_{f_1^{-1}(G)} \diff F_0(a_1) \diff F_0(a_2).
\end{equation*}
Since
\begin{equation*}
\iint\limits_{f_1^{-1}(G)} \diff F_0(a_1) \diff F_0(a_2) = \int_{G} \diff F_0^-(a),
\end{equation*}
the proof is complete.

The second condition \eqref{port2} can be proved in a similar manner.
We will sketch the steps of the proof but leave out the details.
The relevant form of the density evolution equation is now
\begin{align*}
F_n^+(a) & = \iint\limits_{(a_1a_2/a_1\ast a_2) \le a} (a_1\ast a_2) \diff F_n(a_1) \diff F_n(a_2)\\
& + \iint\limits_{(\overline{a}_1a_2/\overline{a}_1\ast a_2) \le a} (\overline{a}_1\ast a_2) \diff F_n(a_1) \diff F_n(a_2).
\end{align*}
We define $f_{21}(a_1,a_2) = a_1a_2/a_1\ast a_2$ and $f_{22}(a_1,a_2) = \overline{a}_1a_2/\overline{a}_1\ast a_2$,
and write
\begin{align*}
P_n^+(G) & \defn \int_{G} \diff F_n^+(a) 
 = \iint\limits_{f_{21}^{-1}(G)} (a_1\ast a_2) \diff F_n(a_1) \diff F_n(a_2) \\
& + \iint\limits_{f_{22}^{-1}(G)} (\overline{a}_1\ast a_2) \diff F_n(a_1) \diff F_n(a_2).
\end{align*}
Next, we note that, by a general result on the preservation of weak convergence \cite[Thm.~5.1]{Billingsley}, 
$$
(a_1\ast a_2) \diff F_n(a_1) \diff F_n(a_2) \Rightarrow (a_1\ast a_2) \diff F_0(a_1) \diff F_0(a_2),
$$
$$
(\overline{a}_1\ast a_2) \diff F_n(a_1) \diff F_n(a_2) \Rightarrow (\overline{a}_1\ast a_2) \diff F_0(a_1) \diff F_0(a_2).
$$
(The important point here is that the functions $(a_1\ast a_2)$ and $(\overline{a}_1\ast a_2)$ are uniformly continuous and bounded over the domain $(a_1,a_2)\in [0,1]^2$.
The claimed convergences follow readily from the definition of weak convergence.)
The proof is completed by writing
\begin{align*}
\liminf_n & P_n^+(G)  \ge  \iint\limits_{f_{21}^{-1}(G)}  (a_1\ast a_2) \diff F_0(a_1) \diff F_0(a_2)\\
 & + \iint\limits_{f_{22}^{-1}(G)}  (\overline{a}_1\ast a_2) \diff F_0(a_1) \diff F_0(a_2) = \int_{G} \diff F_0^+(a).
\end{align*}
\end{proof}

\begin{lemma}\label{lemma:Delta}
For $\delta>0$, $\Delta(\delta) >0$.
\end{lemma}

\begin{proof}
Fix $\delta >0$. The set $\cM_\delta$ can be written as the pre-image of a closed set under a continuous function:
$\cM_\delta = C^{-1}([\delta,M])$. Hence, by a general result about continuity (\cite[4.8]{Rudin}), $\cM_\delta$ is closed;
and, being a subset of the compact set $[0,1]$, it is compact (\cite[2.35]{Rudin}).
Since $C$ is continuous and $\cM_\delta$ is compact, the ``inf'' in \eqref{def:Delta} is achieved by some $F_0 \in \cM_\delta$ (\cite[4.16]{Rudin}):
$\Delta(\delta) = C(F_0)$.
Since $V(F_0)\ge \delta>0$, $F_0$ is not extreme, so by Theorem~\ref{theoremvarentropy2}, $C(F_0)>0$ . 
\end{proof}

\section*{Acknowledgment}
This work was supported in part by the European Commission in the framework of the FP7 Network of Excellence in
Wireless COMmunications NEWCOM\# (contract n.318306). 
Part of this work was done while the author was visiting the Simons Institute for the Theory of Computing, UC Berkeley.


\end{document}